\documentclass{article}

\usepackage{latexsym}

\usepackage{amsmath}

\usepackage{amssymb}

\usepackage{amsthm}

\usepackage{amscd}

\usepackage{graphicx}

\usepackage{subfigure}

\usepackage{psfrag}

\usepackage{bm}

\usepackage{cancel}

\usepackage[mathscr]{eucal}

\newcommand{\textfrac}[2]{{\textstyle \frac{#1}{#2}}}

\newcommand{\R}{\mathbb{R}}

\newcommand{\UM}{U\!M}

\newcommand{\PM}{P\!M}

\newcommand{\fp}{f_{\!\text{\tiny p}}}

\newcommand{\vol}{\mathrm{vol}}

\newcommand{\fourg}{{}^4\mathbf{g}}

\DeclareMathOperator{\tr}{tr}

\theoremstyle{plain}

\newtheorem{Theorem}{Theorem}

\newtheorem{Lemma}{Lemma}

\theoremstyle{remark}

\newtheorem*{Remark}{Remark}

\title{Oscillations toward the singularity of LRS Bianchi type~IX cosmological models with Vlasov matter}

\author{Simone Calogero\footnote{E-Mail: calogero@ugr.es}\\[0.2cm]
Departamento de Matem\'atica Aplicada\\
Facultad de Ciencias, Universidad de Granada\\
18071 Granada, Spain\\[0.5cm]
J.~Mark Heinzle\thanks{E-Mail: Mark.Heinzle@univie.ac.at}\\[0.2cm]
Gravitational Physics\\ Faculty of Physics, University of Vienna\\
1090 Vienna, Austria}

\date { }

\begin{document}

\maketitle

\begin{abstract}
We analyze the dynamics of a class of cosmological
solutions of the Einstein-Vlasov equations.
These equations describe an ensemble of collisionless particles
(which represent galaxies or clusters of galaxies)
that interact gravitatively through Einstein's equations of general relativity.
The cosmological models we consider are spatially homogeneous, of Bianchi type~IX,
and locally rotationally symmetric (LRS).
We prove that generic solutions exhibit an oscillatory approach toward the singularities (the `big bang' in the past and 
the `big crunch' in the future); this is in contrast to
the behavior of Einstein-vacuum or Einstein-Euler solutions.
To establish this result 
we make use of dynamical systems theory; 
we introduce dimensionless dynamical variables that are
defined on a compact state space; in this formulation 
the oscillatory behavior of generic solutions is represented by 
an approach to heteroclinic cycles.
\end{abstract}


\section{Introduction}

\label{intro}


A large number of recent rigorous results on 
the dynamics of cosmological solutions (i.e., models for the `universe')
of the Einstein field equations 
is due to the successful application 
of the theory of dynamical systems to general relativity. 
The Einstein equations form a highly non-linear system of partial differential equations that describe 
the evolution of the metric of the space-time, the latter being represented by a four-dimensional Lorentzian 
manifold. In the case of spatially homogeneous solutions, which are of interest in cosmology, the Einstein 
equations reduce to a system of non-linear {\it ordinary} differential equations, which can be analyzed using 
the powerful methods of dynamical systems theory. 
This is typically achieved by going over from the 
standard metric variables of the Einstein equations 
to dimensionless variables via
conformal rescalings and the introduction
of normalizations to regularize the
equations and to obtain an 
autonomous finite-dimensional dynamical system over a compact state space.
We refer to~\cite{WE} 
for an overview on the theory of dynamical systems in cosmology.



This paper concerns cosmological solutions of the Einstein
equations coupled to collisionless matter (Vlasov matter).
Solutions of the Einstein-Vlasov 
equations represent ensembles of massive particles (like stars in a galaxy)
that interact 
through the gravitational field they create collectively.
In cosmological applications, the particles are thought to represent
galaxies (or galaxy clusters) in the universe.

According to 
the standard theory for the evolution of the universe, 
massless particles (photons) account for most of the energy density in the universe 
during the time of `radiation dominance', which begins at about $1$ second after the big bang and
ends at the time of decoupling between radiation and matter about $10^5$ years after the big bang.
From this time onward, the universe is `matter dominated' and galaxies begin to form 
about $10^6$ years after the big bang. 
It is immediate that the collisionless matter model that describes
ensembles of massive particles is a useful tool to
model structure formation.
The time of `radiation dominance', on the other hand, 
suggests the study of ensembles of massless particles 
(where in the approach to the big bang singularity
collisions are of course expected to play an increasing role---this is not
captured by the collisionless matter model).
Interestingly enough, the dynamics toward the big bang singularity in the case of 
massive particles 
is qualitatively the same as the dynamics of the massless particles model. 
This seems to be a generic property 
of Vlasov matter, see~\cite{HU,RU}. 

The existence of a big bang (an `initial singularity') is predicted 
by the singularity theorems of general relativity~\cite{HE}.
Under very general assumptions,
there will exist causal curves that `terminate' at a singularity, and
physical quantities representing curvature or the energy density
of the matter will diverge along these inextendible curves.
The detailed characterization of these singularities is 
an important problem in general relativity and cosmology.
One interesting question in this respect concerns the
details of the divergence of the curvature toward
the singularity (e.g., along a distinguished congruence of curves 
representing the (spacetime) trajectories of the matter).
The well-known BKL conjecture~\cite{BKL} states that the approach
to the singularity will in general be `oscillatory', which means that
appropriately rescaled (curvature) quantities will oscillate (with
a rapidly increasing frequency) instead of converging monotonically
toward the singularity. The paradigm of this type of behavior
is the so-called `Mixmaster' behavior that originates
from the study of spatially homogeneous spacetimes.
We refer to sections~5 and~6 of the textbook~\cite{WE} for a good overview.
Another important problem concerns the question of whether the asymptotic
dynamics of solutions toward the singularity
are sensitive to the choice of 
matter model or not (i.e., whether ``matter matters'').
Although for matter models like perfect fluids 
the latter seems to be the case,
there are matter models such that the solutions of
the Einstein-matter equations exhibit an asymptotic behavior
that is different from that of vacuum solutions.
The collisionless (Vlasov) matter model is a prime example,
and the results of the present work are results in this vein.

To obtain a detailed characterization of the dynamics of solutions of the Einstein-Vlasov equations,
it is necessary to assume a high degree of symmetry
as, e.g., spatial homogeneity.
The global dynamics of spatially homogeneous 
solutions of the Einstein-Vlasov system has been studied extensively, see, 
e.g.,~\cite{CH3,HU,R,R1,RT,RU} for applications of dynamical system theory
to the problem.
The reformulation of the Einstein(-Vlasov) equations 
as a dynamical system has proved to be highly advantageous. 
The reason for the success of dynamical systems methods lies in the fact that 
the behavior of the spacetime geometry in the neighborhood of an (initial) singularity is
determined by the nature of the $\alpha$-limit set of the dynamical system;
see section~5.3 of \cite{WE} for a good introduction.
In the simplest cases, the $\alpha$-limit set is an isolated fixed point; typically,
this corresponds to the spacetime curvature growing monotonically. If the $\alpha$-limit
set is more complicated, however, for example a heteroclinic cycle with a
finite number of fixed points, then the approach to the singularity will
be `oscillatory'. The paradigm of oscillatory behavior, the Mixmaster behavior described in the BKL conjecture, is
even more intricate---it is induced by infinite heteroclinic chains.

The present work extends the results of~\cite{RT} on massless Vlasov matter
and~\cite{RU} on massive particles. 
The results of these references concern the dynamics of Einstein-Vlasov solutions 
that satisfy the most restrictive symmetry assumptions
(LRS Bianchi type~I,~II and~III);
the analysis is based on techniques from dynamical systems theory 
in conjunction with the use of Hubble-normalized dimensionless variables, see~\cite{WE}.
In the present paper we refine these techniques
to investigate a class of solutions exhibiting a larger number of (true)
degrees of freedom: LRS Bianchi type~IX; we refer to section~\ref{sec:der} for a definition.
Our analysis employs a different set of dimensionless variables to regularize the
equations and to recast the Einstein-Vlasov system into an 
autonomous finite-dimensional dynamical system over a compact state space. 
We note that
the methods we use to prove our main result are closely connected with 
the general formalism developed in~\cite{CH1}.
However, ensembles of massive collisionless particles do not directly fall into the class of matter models 
considered in~\cite{CH1}, which makes it necessary to generalize 
the approach.

The paper is largely self-contained.
In section~\ref{sec:der} we discuss the collisionless matter model
and give a derivation of the Einstein-Vlasov equations for the class
of cosmological models we consider; the symmetry assumptions
(LRS Bianchi type~IX) are explained.
In section~\ref{sec:refor} we reformulate the equations in terms
of dimensionless variables; however, another reformulation of the equations
is necessary to obtain a regular autonomous finite-dimensional dynamical system over a 
compact state space.
The analysis of this dynamical system is performed in section~\ref{result}.
At the end of this section we state the main theorem:
We prove that the $\alpha$- and the $\omega$-limit set of generic orbits of 
the dynamical system that correspond to LRS Bianchi type~IX solutions of the Einstein-Vlasov equations is a heteroclinic cycle.
In the concluding remarks, section~\ref{conc}, we illuminate the main
result from a physical perspective by putting it into a broader context.


\section{Derivation of the equations}

\label{sec:der}


Consider an ensemble of massive particles that are in geodesic motion 
in a `spacetime' $(M, \fourg)$, which is a smooth four-dimensional manifold
equipped with a metric tensor field  of Lorentzian signature $({-+++})$.
The assumption of geodesic motion reflects the condition of absence 
of interactions between the particles other than gravity; since the particles interact solely 
through the gravitational field they create collectively, the governing equations are
Einstein's field equations of general relativity.
This type of matter is commonly called `collisionless matter', since, in particular, interactions by 
collisions are excluded.
Examples of physical systems that are believed to be well approximated by 
the collisionless matter model in gravity 
are galaxies or galaxy clusters; in the former case, the 
particles are the stars of the galaxy, while in the cosmological setting, 
the particles are the galaxies of the cluster~\cite{And, And2, galactic, RR}.

The ensemble of particles is represented by a `particle distribution function' 
$\fp: \mathbb{R}^+ \times \UM \rightarrow [0,\infty)$,
where $\UM\subset TM$ is the bundle of unit mass shells (four-velocity hyperboloids),
i.e., the subset of the tangent bundle $TM$ given by the condition
$\fourg(u,u) = {-1}$ (where $u$ is a future-directed four-velocity).
Let $x\in M$ and $u\in \UM$ be a four-velocity at $x$; let
$\vol_{\UM}$ denote the induced volume element on $\UM$; 
then
\[
\fp(m,x,u)\, \vol_{\UM}\, d m
\]
represents the proper number density of those particles whose mass is 
in an interval of (infinitesimal) length $d m$ around $m$ and whose four-velocity is in an (infinitesimal) 
volume $\vol_{\UM}$ containing $u$.
The proper (rest) mass density of the particles whose four-velocity is in a 
volume $\vol_{\UM}$ containing $u$ is 
\begin{equation}\label{massdens}
f(x,u)\, \vol_{\UM} = \Big( \int_{\mathbb{R}^+} m \fp(m,x,u) \,d m \Big) \:\vol_{\UM} \:;
\end{equation}
this relation defines the `mass distribution function' $f$.

Let $(t,x^i)$ be a system of coordinates on $M$ 
and $\{e_0,e_i\}$ be a frame, e.g., the coordinate frame
$\{\partial_t, \partial_{x^i}\}$; we assume that $e_0$ is (future-directed) timelike
and $e_i$ spacelike, $i=1,2,3$.
Then the spatial components $u^i$ of the four-velocity (w.r.t.\ the frame)
are coordinates on the hyperboloid $\UM$ and we can express the invariant measure on $\UM$ 
as $\vol_{\UM} = \sqrt{|\det \fourg|}\: |u_0|^{-1} \,du^1 d u^2 d u^3$,
where $u_0$ is determined from $u^i$
by the normalization relation $\fourg(u,u) = g_{\mu\nu} u^\mu u^\nu = {-1}$.
We adhere to the convention that spacetime indices are denoted by Greek letters, whose range is $0,1,2,3$, while
Latin indices are spatial indices and take the values $1,2,3$.
We use the Einstein summation convention.

Regarding $f$ (and $\fp$) as functions of the coordinates $t$, $x^i$, $u^j$ (and $m$), $i,j=1,2,3$,
we find that the energy-momentum tensor of the ensemble is given by
\begin{equation}\label{Tvlasov}
T^{\mu\nu} = \int\vol_{\UM}\int dm \,m \,\fp\, u^\mu u^\nu =
\int\vol_{\UM}\, f\, u^\mu u^\nu = 
\int f\, u^\mu u^\nu \, \sqrt{|\det \fourg|}\: |u_0|^{-1} \,du^1 d u^2 d u^3\:,
\end{equation}
where $u_0 = u_0(u^1,u^2,u^3)$ by 
the condition $\fourg(u,u) = g_{\mu\nu} u^\mu u^\nu = {-1}$.

\begin{Remark}
A common special case is the case where the ensemble of particles consists of one species
of particles with equal mass $\mathsf{m}$, which corresponds to
$\fp(m,x,u) = \delta(m-\mathsf{m}) \,\bar{\mathsf{f}}(x,u)$. 
In this context it is customary to use 
the four-momentum $v = \mathsf{m} u$ as the variable of the distribution function;
taking the relation between the volume elements on the unit mass shell 
$\UM$ and the mass shell $\PM = \{ v \,|\, \fourg(v,v) ={-\mathsf{m}^2} \}$ into account, we see that
$\mathsf{f}(x,v) \vol_{\PM}$ with $\mathsf{f}(x,v) = \mathsf{m}^{-3}\, \bar{\mathsf{f}}(x,u)$ represents the proper number density.
Then $f(x,u) = \mathsf{m}^4 \,\mathsf{f}(x,v)$ and~\eqref{Tvlasov}
becomes
\begin{equation*}
T^{\mu\nu} = \mathsf{m}^{-4}\, \int f\big(x, \textfrac{v}{\mathsf{m}}\big)\, v^\mu v^\nu \, \sqrt{|\det \fourg|}\: \frac{d v^1 d v^2 d v^3}{|v_0|}
= \int \mathsf{f}(x,v)\, v^\mu v^\nu \, \sqrt{|\det \fourg|}\: \frac{d v^1 d v^2 d v^3}{|v_0|} \:,
\end{equation*}
where $v_0$ is determined from $v^i$
by the mass shell relation $\fourg(v,v) = g_{\mu\nu} v^\mu v^\nu = {-\mathsf{m}^2}$.
\end{Remark}

\begin{Remark}
The formalism to describe ensembles of massless particles is analogous. In the massless case, 
the domain of the distribution function is the bundle of future light cones, i.e., the
set of future-directed null vectors.
Accordingly, the energy-momentum tensor is given by~\eqref{Tvlasov}, 
where $u_0$ is determined 
from $u^i$ by the condition $\fourg(u,u) = g_{\mu\nu} u^\mu u^\nu = 0$.
\end{Remark}

Both the particle distribution function $\fp$ and the mass distribution function $f$ 
satisfy the Vlasov equation
\begin{equation}\label{vlasoveq}
\partial_t f +\frac{u^j}{u^0}\partial_{x^j}f-\frac{1}{u^0}\Gamma^j_{\mu\nu}u^\mu u^\nu\partial_{u^j}f=0\:,
\end{equation}
which reflects the condition of geodesic motion of the particles; $\Gamma^\mu_{\nu\sigma}$ are the Christoffel symbols. 
Note in particular that the characteristic curves of the Vlasov equation, along which $f$ is constant, 
coincide with the lift on $\UM$ of the spacetime geodesics. The gravitational interaction
of the particles is modeled by the Einstein-Vlasov system, i.e.,
by coupling~\eqref{vlasoveq} to the Einstein equations of general relativity,
\begin{equation}\label{einsteineqs}
R_{\mu\nu}-\frac{1}{2}g_{\mu\nu}R=T_{\mu\nu}\:.
\end{equation}
In these equations, $R_{\mu\nu}$ is the Ricci tensor of the metric $\fourg$ 
and $R=g^{\mu\nu}R_{\mu\nu}$ the Ricci scalar;
for $T_{\mu\nu}$ we use the energy-momentum tensor~\eqref{Tvlasov} representing the
Vlasov matter.
We adopt units such that $8\pi G=c=1$, where $G$ is Newton's gravitational constant and $c$ the speed of light.

In this paper 
we consider spatially homogeneous spacetimes of Bianchi type~IX
that are locally rotationally symmetric (LRS), i.e., spacetimes
of the form $M=I\times S^3$ (where $I$ is an interval of $\mathbb{R}$)
with metric 
\begin{equation}\label{metric}
\fourg= -dt^2 + g_{11}(t)\:\hat{\omega}^1\otimes\hat{\omega}^1+
g_{22}(t)\,(\hat{\omega}^2\otimes\hat{\omega}^2+\hat{\omega}^3\otimes\hat{\omega}^3)\:,
\end{equation}
where $\{\hat{\omega}^1,\hat{\omega}^2,\hat{\omega}^3\}$ 
is a time-independent coframe on $S^3$ that satisfies 
$d\hat{\omega}^1=-\hat{\omega}^2\wedge\hat{\omega}^3$ (and cyclic permutations). 
As proved in~\cite{MM}, the general solution of the Vlasov equation~\eqref{vlasoveq} on a background spacetime with 
metric~\eqref{metric} can be expressed as
\begin{equation}\label{generalf}
f=f_0(u_1,(u_2)^2+(u_3)^2)\:,
\end{equation} 
where $u_i = g_{i j} u^j$, $i=1,2,3$, and 
$f_0:\R\times\R_+\to\R_+$ is an arbitrary sufficiently smooth function.

Let us compute~\eqref{Tvlasov} from~\eqref{generalf}.
Denoting by $g$ the spatial Riemannian metric we have \mbox{$|\det \fourg | = \det g$}
and we find $d u^1 d u^2 d u^3 = (\det g)^{-1} d u_1 d u_2 d u_3$; moreover,
$|u_0|^2 = 1 + g^{11} u_1^2 + g^{22} (u_2^2 + u_3^2)$.
Therefore, the energy density $\rho= T_{00}$ is given by
\begin{subequations}\label{compT}
\begin{align}\label{rho}
\rho & =(\det g)^{-1/2}\:\int f_0\, \big(1+g^{11}u_1^2+g^{22}(u_2^2+u_3^2)\big)^{1/2}\;du_1du_2du_3
\intertext{and the principal pressures $p_1=T^1_{\ 1}$, $p_2=T^2_{\ 2}$, $p_3=T^3_{\ 3}$ are}
\label{pi}
p_i & =(\det g)^{-1/2} \:\int f_0 \,g^{ii}\, u_i^2\, \big(1 +g^{11} u_1^2+g^{22}(u_2^2+u_3^2)\big)^{-1/2}\;du_1du_2du_3\:,
\end{align}
\end{subequations}
where there is no summation over $i$.
Since $f_0$ is the function~\eqref{generalf},
we find $p_2=p_3$.

\begin{Remark}
The energy density and the principal pressures~\eqref{compT} depend on the 
arbitrary function $f_0$. 
This function can be interpreted as the `initial data' for $f$ at some time $t=t_0$,
because $f(t_0, u^1, u^2, u^3) = f_0\big( g_{1 1}(t_0) u^1, (g_{2 2}(t_0))^2 ((u^2)^2+(u^3)^2)\big)$.
Once the initial data $f_0$ is prescribed, 
$\rho$ and $p_1$, $p_2$ are functions of the metric components,
which can be interpreted as 
implicit relations between the principal pressures and the energy density, i.e., 
as an `equation of state'.
To compare, let us briefly recall the perfect fluid matter model.
In the perfect fluid case, once an equation of state $p = p(\rho)$ is
prescribed, the energy density and the pressure are determined
by the constraints.
The evolution equations for the matter---the Euler equations---are equivalent to the 
conservation of the energy momentum tensor, $\nabla_\mu T^{\mu\nu}=0$, and 
are thus contained in the Einstein equations~\eqref{einsteineqs} 
(through the Bianchi identities).
In contrast, the Vlasov equation~\eqref{vlasoveq} 
is independent; initial data 
for the first order 
equation~\eqref{vlasoveq} has to be prescribed 
in order to obtain a solution.
\end{Remark}

For the pair~\eqref{metric}--\eqref{generalf} 
to be a candidate for a solution of the Einstein-Vlasov system, 
the energy momentum tensor must be compatible with the LRS assumption,
i.e., diagonal and $T^2_{\ 2}=T^3_{\ 3}$. 
This can be achieved by restricting to distribution functions~\eqref{generalf} 
that are invariant under the transformation $u_1\to -u_1$. These distribution 
functions are called {\it reflection symmetric}, see~\cite{R}. 
Besides reflection symmetry, for technical reasons, we also assume 
that $f_0$ has {\it split support}, which means that 
the support of $f_0$ does not intersect any of the axes.


\begin{Remark}
The (rest) mass current density of particles 
is given as
\[
N^{\mu}=(\det g)^{-1/2} \: \int f_0 \,u^\mu \: |u_0|^{-1} \,d u_1 d u_2 d u_3 
\]  
and satisfies $\nabla_\mu N^\mu=0$, which expresses the conservation of (rest) mass.
A straightforward consequence of the assumption of reflection symmetry is that $N^i=0$, i.e., 
the current density is orthogonal to the hypersurfaces $t=const$. 
The matter model can thus interpreted as being `non-tilted'.
\end{Remark}

The Einstein equations~\eqref{einsteineqs} split into the Hamiltonian constraint
equation
\begin{subequations}\label{einstein}
\begin{equation}\label{constraint}
9 H^2 - \big((k^1_{\ 1})^2+2(k^2_{\ 2})^2\big)+ R = 2 \rho \:,
\end{equation}
and the evolution equations
\begin{align}
\label{metrolution}
\partial_t g^{ii} & = 2 k^i_{\ i} \,g^{ii}\:, \\
\label{evolution}
\partial_t k^i_{\ i}  & = R^i_{\ i} -3 H k^i_{\ i}-p_i+\frac{1}{2}\left(p_1+2p_2-\rho\right)\:,
\end{align}
where there is no summation over $i$.
The quantities $k^1_{\ 1}$ and $k^2_{\ 2} = k^3_{\ 3}$ are the components 
of the second fundamental form of the hypersurfaces $t=\mathrm{const}$, 
the Hubble scalar $H$ is defined by $H = -{\textfrac{1}{3}}\, \tr k$.
\begin{equation}\label{Rs}
R^1_{\ 1}=\frac{1}{2}\frac{(g^{22})^2}{g^{11}}\:,\qquad R^2_{\ 2}=R^3_{\ 3}=g^{22}-\frac{1}{2}\frac{(g^{22})^2}{g^{11}}
\end{equation}
\end{subequations}
are the non-zero components of the Ricci tensor, and $R=R^1_{\ 1}+2R^2_{\ 2}$ is the Ricci scalar.
There is another constraint equation, the momentum constraint, that is satisfied identically by our assumptions.

It is well known that the maximal interval of existence of solutions of the 
Einstein-Vlasov system~\eqref{compT}--\eqref{einstein} is of the form $(t_-,t_+)$, where $|t_\pm|<\infty$; 
we refer to~\cite{R1,R2}.
After a time translation we may assume that the singularity in the past is at $t_-=0$. 
The Einstein vacuum equations are recovered from~\eqref{compT}--\eqref{einstein}
by setting $f_0 \equiv 0$. 
The asymptotic behavior of vacuum solutions 
is characterized by the existence of positive constants $a_\pm$, $b_\pm$, such that
\begin{equation}\label{vacasy}
g_{11}(t)= a_\pm (t-t_\pm)^2 \,\big( 1 + o(1) \big) \:,\qquad g_{22}(t)=g_{33}(t)= b_\pm + o(1) \qquad (t\rightarrow t_\pm)\:.
\end{equation}
This behavior is not exclusive to vacuum solutions; on the contrary,~\eqref{vacasy}
is the typical behavior of solutions associated with a large variety of matter sources,
(non-stiff) perfect fluids being the prime example~\cite{CH1}.
Since the metric
\begin{subequations}
\begin{align}\label{taub}
& \mathrm{T}: \quad -dt^2 + a \,(t-t_\pm)^2\, (d x^1)^2 + b\, \big( (d x^2)^2 + (d x^3)^2 \big) 
\intertext{with $a, b = \mathrm{const}$ 
is the so-called Taub solution (flat LRS Kasner solution),
one refers to~\eqref{vacasy} as an approach to the Taub solution.
(However, the spacetime associated with the metric~\eqref{taub} is $(0,\infty) \times T^3$ 
instead of $I \times S^3$ and thus of Bianchi type~I; the Taub solution does not
satisfy~\eqref{Rs}.)
The second class of vacuum LRS solutions of Bianchi type~I (i.e., on $(0,\infty) \times T^3$)
is the class of non-flat LRS solutions} 
\label{solQ}
& \mathrm{Q}: \quad -dt^2 + a \,(t-t_\pm)^{-2/3}\, (d x^1)^2 + b\,(t-t_\pm)^{4/3}\, \big( (d x^2)^2 + (d x^3)^2 \big) \:;
\end{align}
\end{subequations}
as opposed to~\eqref{taub}, 
the class $\mathrm{Q}$ does not play any particular role
in the context of the asymptotic dynamics of vacuum solutions (or perfect fluid solutions)
of~\eqref{compT}--\eqref{einstein}.

In this paper we are interested in 
the asymptotic behavior of solutions of the Einstein-Vlasov system~\eqref{compT}--\eqref{einstein} 
as $t\rightarrow t_\pm$. 
Our main result can be informally stated as follows: 
{\it In the limit $t\to t_\pm$, generic solutions of the system~\eqref{compT}--\eqref{einstein}
oscillate between the Taub class and the class of non-flat LRS Kasner solutions.}
The asymptotic behavior of solutions of the Einstein-Vlasov system is thus
qualitatively different from that of vacuum solutions.
Therefore, ``collisionless matter matters'', as opposed to (non-stiff) perfect fluid matter 
and several other matter sources which do not affect the
structure of the singularity.

\begin{Remark}
The 
fact that the behavior toward the singularity of 
cosmological models with collisionless matter 
is qualitatively different from that of vacuum and perfect fluid models
is known from LRS Bianchi type~II models~\cite{RT,RU}.
The occurrence of oscillations in the asymptotic dynamics 
of LRS models 
that are induced by the anisotropy of the matter model
has subsequently been studied in some detail, e.g., in~\cite{CH1, CH5}.
\end{Remark}


\section{Reformulated Einstein-Vlasov equations}

\label{sec:refor}


Let 
\[
n = \frac{1}{\sqrt{\det g}} = \sqrt{g^{11} (g^{22})^2}
\quad\text{ and }\quad
s=\frac{g^{22}}{g^{11} + 2 g^{22}} \:.
\]
Note 
that due to~\eqref{massdens} and~\eqref{generalf}, 
$n = n(t) \in (0,\infty)$ is proportional to the mass (or particle) density
of the ensemble of particles (as measured w.r.t.\ the frame associated with~\eqref{metric}).
The variable $s$, on the other hand, satisfies 
$s = s(t) \in (0,\textfrac{1}{2})$ and is a (non-linear) measure of
the deviation of the metric from isotropy; $s = \textfrac{1}{3}$ ($\Leftrightarrow g_{11}=g_{22}$) corresponds to an isotropic metric.
We further define
\[
\ell = \frac{1}{1 + n^{2/3}} \;.
\]
Since $n^{2/3} = (\det g)^{-1/3}$, $n^{2/3}$ corresponds to a length scale of the metric,
and $\ell$ is a non-linear measure of such a scale; $\ell = 0$
corresponds to a singularity ($\det g = 0$); $\ell = 1$
to a state of infinite volume ($\det g = +\infty$).
The equation satisfied by $\ell$ follows directly from~\eqref{metrolution},
\begin{equation}
\partial_t \ell = 2 H \ell (1-\ell) \:;
\end{equation}
recall that $H = {-\textfrac{1}{3}}\,\tr k$ is the Hubble scalar; accordingly, $H > 0$ means 
expansion, $H< 0$ contraction.
For $s$ we find $\partial_t s = -2 s (1- 2s) (k^1_{\ 1} - k^2_{\ 2})$,
where $k^1_{\ 1} - k^2_{\ 2}$ can be identified with (three times) the $2$-$2$-component
of the shear tensor.

We are able to express the energy density~\eqref{rho} and the principal pressures~\eqref{pi}
in terms of $\ell$ and $s$; defining
\begin{equation}\label{wiw}
w_i = \frac{p_i}{\rho}\:,\qquad
w = \frac{p}{\rho} = \frac{1}{3}\, \frac{p_1 + 2 p_2}{\rho} = \frac{1}{3} \,\big(w_1 + 2 w_2\big)
\end{equation}
we obtain 
\begin{subequations}\label{wi}
\begin{align}
w_1 &= (1-\ell) (1-2s)
\,\frac{\int f_0 \,u_1^2\, \left[\ell \big(s^2 (1-2s)\big)^{1/3}+ (1-\ell)\big((1-2s)u_1^2+s(u_2^2+u_3^2)\big)\right]^{-1/2}\,du}%
{\int f_0\left[\ell \big(s^2 (1-2s)\big)^{1/3}+(1-\ell)\big((1-2s)u_1^2+s (u_2^2+u_3^2)\big)\right]^{1/2}du}\:,\\[1ex]
w_2 &= (1-\ell) \, s\; \frac{\int f_0 \,u_2^2 \,
\left[\ell \big(s^2 (1-2s)\big)^{1/3}+ (1-\ell)\big((1-2s)u_1^2+s(u_2^2+u_3^2)\big)\right]^{-1/2}\,du}%
{\int f_0\left[\ell \big(s^2 (1-2s)\big)^{1/3}+(1-\ell)\big((1-2s)u_1^2+s (u_2^2+u_3^2)\big)\right]^{1/2}du}\:,
\end{align}
\end{subequations}
where $du$ abbreviates $du_1du_2du_3$.
The assumption of split support assures that $w_1$, $w_2$ are 
smooth functions, since the denominator in~\eqref{wi} is strictly positive for all values of $\ell$ and $s$.

For our analysis it is necessary to recast the Einstein equations~\eqref{einstein}
into a different form.
We use the dominant variable
\begin{subequations}\label{newvar}
\begin{equation}
D=\sqrt{H^2+\frac{1}{3}\,g^{22}} = \sqrt{\frac{1}{9} \,(\tr k)^2 + \frac{1}{3} \,g^{22}}\:
\end{equation}
to define normalized dimensionless variables
according to 
\begin{equation}
H_D=\frac{H}{D}\:,\qquad
\Sigma_+=\frac{k^1_{\ 1}-k^2_{\ 2}}{3D}\:,
\qquad 
M_1=\frac{1}{D}\frac{g^{22}}{\sqrt{g^{11}}}\:,
\qquad
\Omega=\frac{\rho}{3D^2}\:.
\end{equation}
In addition we replace the cosmological time $t$ by a rescaled time variable $\tau$ via
\begin{equation}\label{newtime}
\frac{d}{d\tau}=\big(\quad\big)^\prime =\frac{1}{D}\frac{d}{dt}\:.
\end{equation}
\end{subequations}
Rewriting the Einstein equations~\eqref{einstein} in the new variables, 
we obtain a decoupled ODE for $D$
and a system of ODEs that we call the \textit{reduced dynamical system}:
\begin{subequations}\label{dynsyst}
\begin{align}
\label{HDEq2}
H_D^\prime & = -(1-H_D^2) (q - H_D \Sigma_+) \:,\\[0.5ex]
\label{Sig+Eq}
\Sigma_+^\prime & = -(2- q) H_D\Sigma_+ - (1-H_D^2) (1-\Sigma_+^2) + \frac{1}{3}\,  M_{1}^2 
+ \Omega \,\big(w_2(\ell,s) -w_1(\ell,s) \big) \:,\\[0.5ex]
\label{M1Eq}
M_{1}^\prime & = M_{1} \big( q H_D - 4 \Sigma_+ + (1-H_D^2) \Sigma_+ \big)\:,\\[1.5ex]
\label{zeq}
\ell'&=2 H_D \ell (1-\ell )\:.
\end{align}
\end{subequations}
In the system~\eqref{dynsyst}, $q$ is the so-called 
deceleration parameter, $q=2\Sigma_+^2+\frac{1}{2}(1+3w)\Omega$;
in addition, $\Omega$ is determined from the variables $\Sigma_+$ and $M_1$ 
by the Hamiltonian constraint~\eqref{constraint},
and~\eqref{newvar} is used to express $s$ as a function of $H_D$ and $M_1$, i.e.,
\begin{equation}\label{omega}
\Omega=1-\Sigma_+^2-\frac{1}{12}M_1^2\:,\qquad
s = \Big( 2 + \frac{3 (1- H_D^2)}{M_1^2}\Big)^{-1} \:.
\end{equation}

The system~\eqref{dynsyst} is a closed system that completely describes the
dynamics of LRS Bianchi type~IX Einstein-Vlasov models.
Note that the r.h.s.\ of~\eqref{dynsyst} contains the functions~\eqref{wi}
that are determined by an integration over $f_0 = f_0(u_1,u_2,u_3)$.
This means that, in particular, the r.h.s.\ of the dynamical system~\eqref{dynsyst}
depends on the initial data $f_0$, which is an interesting feature of the problem.
A more detailed derivation of~\eqref{dynsyst} is given in~\cite{CH1}.

\begin{Remark}
The dynamical system~\eqref{dynsyst} is invariant under the discrete symmetry
\[
\tau\to -\tau\:,\quad H_D\to -H_D\:,\quad\Sigma_+\to -\Sigma_+\:.
\]
Hence the qualitative behavior of solutions in the limit $\tau\to +\infty$ mirrors the 
behavior at $\tau\to -\infty$, and we may thus restrict ourself to study the latter. 
\end{Remark}

The state space for the dynamical system~\eqref{dynsyst} is given by
\[
\mathcal{E}_{\mathrm{IX}} = \mathcal{X}_{\mathrm{IX}} \times (0,1)\:,\:
\text{ where }
\mathcal{X}_{\mathrm{IX}} = \Big\{ (H_D, \Sigma_+, M_1) \:\big|\: H_D \in (-1,1)\,,\: M_1 > 0\,,\: 
\Sigma_+^2 + \frac{1}{12} M_1^2 < 1 \Big\}\:,
\]
see Fig.~\ref{b9state}.
The set $\mathcal{E}_{\mathrm{IX}}$ is relatively compact; 
the system~\eqref{dynsyst} is smooth on $\mathcal{E}_{\mathrm{IX}}$. 
However,~\eqref{dynsyst} does not admit a regular extension to 
the entire boundary $\partial\mathcal{E}_{\mathrm{IX}}$, which is because 
the variable $s$ does not have a well-defined limit when $M_1\to 0$ and $H^2_D\to 1$ 
simultaneously. This defect will be remedied by introducing the equivalent system~\eqref{dynsyspolar}.

\begin{Remark}
Let us briefly comment on the equations describing the dynamics of cosmological models
with different matter sources.
The reduced dynamical system for perfect fluid matter is obtained from~\eqref{dynsyst} 
by formally setting \mbox{$w_1 = w_2 = w = \mathrm{const}$} in~\eqref{Sig+Eq}, 
which reflects the isotropy of the matter model
(and the assumption of a linear equation of state).
The equation for $\ell$ decouples from the 
remaining equations and the reduced dynamical system becomes the set of equations~\eqref{HDEq2}--\eqref{M1Eq} 
on the state space $\mathcal{X}_\mathrm{IX}$. 
The reduced dynamical system in the case of an ensemble of massless particles 
is characterized by a decoupling of the equation for $\ell$ as well.
This is because, in the massless case, the renormalized principal pressures
are obtained from~\eqref{wi} by formally setting $\ell=0$.
The reduced dynamical system thus becomes the set of equations~\eqref{HDEq2}--\eqref{M1Eq}.
We therefore find that $\mathcal{X}_\mathrm{IX}$ is the state space 
for both the perfect fluid case and the massless Vlasov case.
We remark that this is not so for the lower Bianchi types. As shown in~\cite{RT}, the state space for 
massless Vlasov particles has one dimension more than the state space for perfect fluids when the Bianchi 
type is I, II or III. In the Bianchi type~IX case, the state spaces are identical, 
but another difference occurs:
While the reduced dynamical system for perfect fluids admits a smooth extension to 
the boundary of $\mathcal{X}_\mathrm{IX}$, 
the Vlasov case is defective in this respect. Loosely speaking, the dynamics of Vlasov matter for massless 
particles does not live naturally in the state space $\mathcal{X}_\mathrm{IX}$.   
\end{Remark}


\section{Basic lemmas}

\label{sec:basiclemmas}


We use the system~\eqref{dynsyst} to prove two basic lemmas.

\begin{figure}[Ht]
\begin{center}
\psfrag{-1}[cc][cc][0.7][0]{$-1$}
\psfrag{1}[cc][cc][0.7][0]{$1$}
\psfrag{0}[cc][cc][0.7][0]{$-1$}
\psfrag{12}[cc][cc][0.7][0]{$1$}
\psfrag{sig}[cc][cc][0.8][0]{$\Sigma_+$}
\psfrag{m1}[cc][cc][0.8][0]{$M_1$}
\psfrag{ss}[cc][cc][1][0]{$ $}
\psfrag{v}[cc][cc][1][0]{$ $}
\psfrag{xi-}[cc][cc][1][0]{$ $}
\psfrag{xi+}[cc][cc][1][0]{$ $}
\psfrag{s}[cc][cc][0.8][0]{$H_D$}
\psfrag{l}[cc][cc][1][-40]{$ $}
\psfrag{h0}[cc][cc][0.8][-45]{$\mathbf{H_D=0}$}
\includegraphics[width=0.5\textwidth]{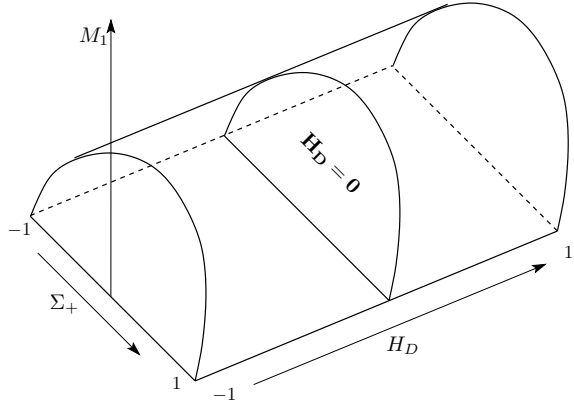}
\end{center}
\caption{The set $\mathcal{X}_\mathrm{IX}$.}
\label{b9state}
\end{figure}

\begin{Lemma}\label{HDLemma}
For every Bianchi type~IX solution with Vlasov matter
there exists $\tau_0 \in \mathbb{R}$ such that
\begin{itemize}
\item $H_D(\tau) > 0$ $\forall \tau < \tau_0$ and $H_D(\tau)$ is bounded away from zero as $\tau\rightarrow {-\infty}$;
\item $H_D(\tau_0) = 0$;
\item $H_D(\tau) < 0$ $\forall \tau > \tau_0$ and $H_D(\tau)$ is bounded away from zero as $\tau\rightarrow \infty$.
\end{itemize}
\end{Lemma}

\begin{proof}
The average pressure $p$ of collisionless matter is non-negative; eqs.~\eqref{wiw} and~\eqref{wi}
imply \mbox{$w = p/\rho \geq 0$}. Therefore the general result by Lin and Wald~\cite{LW} applies:
Every Bianchi type~IX model (with collisionless matter) possesses an initial singularity (`big bang'),
expands initially, then reaches a time when the spatial volume is maximal, 
and finally recontracts to terminate in a singularity (`big crunch').
Hence there exists $\tau_0$ such that $H_D(\tau) > 0$ $\forall \tau < \tau_0$, $H_D(\tau_0) = 0$,
and $H_D(\tau) < 0$ $\forall \tau > \tau_0$. It remains to prove that 
there exists a positive $\epsilon$ such that $H_D(\tau) \geq \epsilon$ ($H_D(\tau)\leq {-\epsilon}$)
for all sufficiently small (large) $\tau$.
To do so assume the contrary, i.e., the existence of a solution whose
$\alpha$-limit set has a non-empty intersection with the plane $H_D = 0$.
We first note that 
\[
{H_D'}_{\,|H_D=0}=-2\Sigma_+^2-\frac{1}{2}(1+3w)\Omega\:,
\]
which is negative unless $\Omega = 0$ and $\Sigma_+ = 0$.
Second, 
\[
{H_D'}_{\,|H_D=0, \Omega = 0,\Sigma_+ = 0}=0 \:,
\qquad
{H_D^{\prime\prime}}_{\,|H_D=0, \Omega = 0,\Sigma_+ = 0}=0 \:,
\qquad
{H_D^{\prime\prime\prime}}_{\,|H_D=0, \Omega = 0,\Sigma_+ = 0} ={-36} < 0 \:.
\]
Let $\mathrm{P}$ be an $\alpha$-limit point of the solution such that $\mathrm{P}$ lies on $H_D = 0$.
Together with $\mathrm{P}$, the entire orbit through $\mathrm{P}$ is contained
in the $\alpha$-limit set; hence there exists a sequence of times $(\tau_n)_{n\in\mathbb{N}}$
with $\tau_n \rightarrow -\infty$ ($n\rightarrow \infty$)
such that $H_D(\tau_n -\delta)>0$, $H_D(\tau_n) = 0$, $H_D(\tau_n +\delta)< 0$
for a sufficiently small $\delta$;
a contradiction.
\end{proof}

\begin{Remark}
A posteriori, by Lemma~\ref{BianchiIXtheo}, we find $H_D(\tau)\rightarrow \pm 1$ as $\tau\rightarrow \mp\infty$.
A note of caution: 
There exist anisotropic matter models other than collisionless matter such
Lemma~\ref{HDLemma} is false; for such matter models there exist (typical) solutions
that do not recollapse but expand forever (i.e., $H_D(\tau) > 0$ $\forall \tau$);
we refer to~\cite{CH4}.
\end{Remark}

\begin{Lemma}\label{ellLemma}
Let\/ $\mathrm{P}$ be an $\alpha$-limit point of a Bianchi type~IX solution with Vlasov matter
as represented by an orbit of~\eqref{dynsyst}.
Then 
\[
\ell_{| \mathrm{P}} = 0 \:.
\]
\end{Lemma}

\begin{proof}
The result is a simple consequence of eq.~\eqref{zeq} for $\ell$ 
and Lemma~\ref{HDLemma}.
\end{proof}

\begin{Remark}
Lemma~\ref{ellLemma} has a straightforward physical interpretation.
Since~\eqref{dynsyst} with $\ell = 0$ describes the dynamics of solutions
of the Einstein-Vlasov equations with massless particles, see the remark
at the end of section~\ref{sec:refor}, we find that the asymptotic dynamics toward the past (and future) singularity 
in the massive Vlasov case are governed by the 
the dynamics of the massless Vlasov case.
\end{Remark}


\section{Analysis and main result}

\label{result}


By Lemma~\ref{HDLemma}, the subset $H_D > 0$ of the state space 
$\mathcal{E}_{\mathrm{IX}}$ is past-invariant
under the flow of~\eqref{dynsyst}; this makes it possible to 
transform~\eqref{dynsyst}
to a different system of equations that is equivalent to~\eqref{dynsyst} on the subset $H_D > 0$.

Let 
\begin{subequations}\label{polartransf}
\begin{align}
M_1^2 & =  3 \,r \sin \vartheta\;, \\[0.2ex]
\label{1minusHD2}
(1 - H_D^2) & = 2 \,r \cos\vartheta\;, \\[0.2ex]
\Sigma_+ & =\text{unchanged}\:,
\end{align}
\end{subequations}
where we assume that $(H_D, M_1,\Sigma_+) \in\mathcal{X}_\mathrm{IX}^+ = \mathcal{X}_\mathrm{IX} \cap \{H_D > 0\}$.
The transformation of variables from $(H_D, M_1,\Sigma_+)$ to $(r, \vartheta,\Sigma_+)$, 
where $r>0$ and $0<\vartheta< \frac{\pi}{2}$ is required,
is a diffeomorphism on the domain $(H_D, M_1,\Sigma_+) \in \mathcal{X}_\mathrm{IX}^+$.
We define the domain $\mathcal{Y}_{\mathrm{IX}}^+$ of the variables $(r,\vartheta,\Sigma_+)$ to be 
the preimage of that domain under the transformation~\eqref{polartransf}.
We obtain $r \cos\vartheta < \frac{1}{2}$ from~\eqref{1minusHD2}; the constraint
$1 - \Sigma_+^2 - \frac{1}{12}\,M_1^2 = \Omega > 0$ implies
$r\sin\vartheta < 4 (1 -\Sigma_+^2)$. Therefore, $\mathcal{Y}_{\mathrm{IX}}^+$ can be written
as 
\begin{equation}\label{Yss}
\mathcal{Y}_{\mathrm{IX}}^+ = \Big\{ r>0, \vartheta \in\big(0,\frac{\pi}{2}\big), \Sigma_+ \in (-1,1)\:\big|\:
r < \min\Big[\frac{1}{2 \cos\vartheta},\frac{4(1-\Sigma_+^2)}{\sin\vartheta}\Big] \:\Big\}\:,
\end{equation}
see Fig.~\ref{Yspace}.

\begin{figure}[Ht]\label{Y}
\begin{center}
\psfrag{sigma}[cc][cc][0.7][0]{$\Sigma_+$}
\psfrag{r}[cc][cc][0.7][0]{$r$}
\psfrag{theta}[cc][cc][0.7][0]{$\vartheta$}
\psfrag{theta}[cc][cc][0.7][0]{$\vartheta$}
\psfrag{L}[cc][cc][1][0]{\eqref{r0}}
\psfrag{B}[cc][cr][1][0]{\ \ \eqref{vartheta0}}
\psfrag{S}[cc][cc][1][0]{\ \ \eqref{varthetapi2}}
\psfrag{h0}[cl][cl][1][-45]{$H_D=0$}
\subfigure[The space $\mathcal{Y}_{\mathrm{IX}}^+$]{\includegraphics[width=0.4\textwidth]{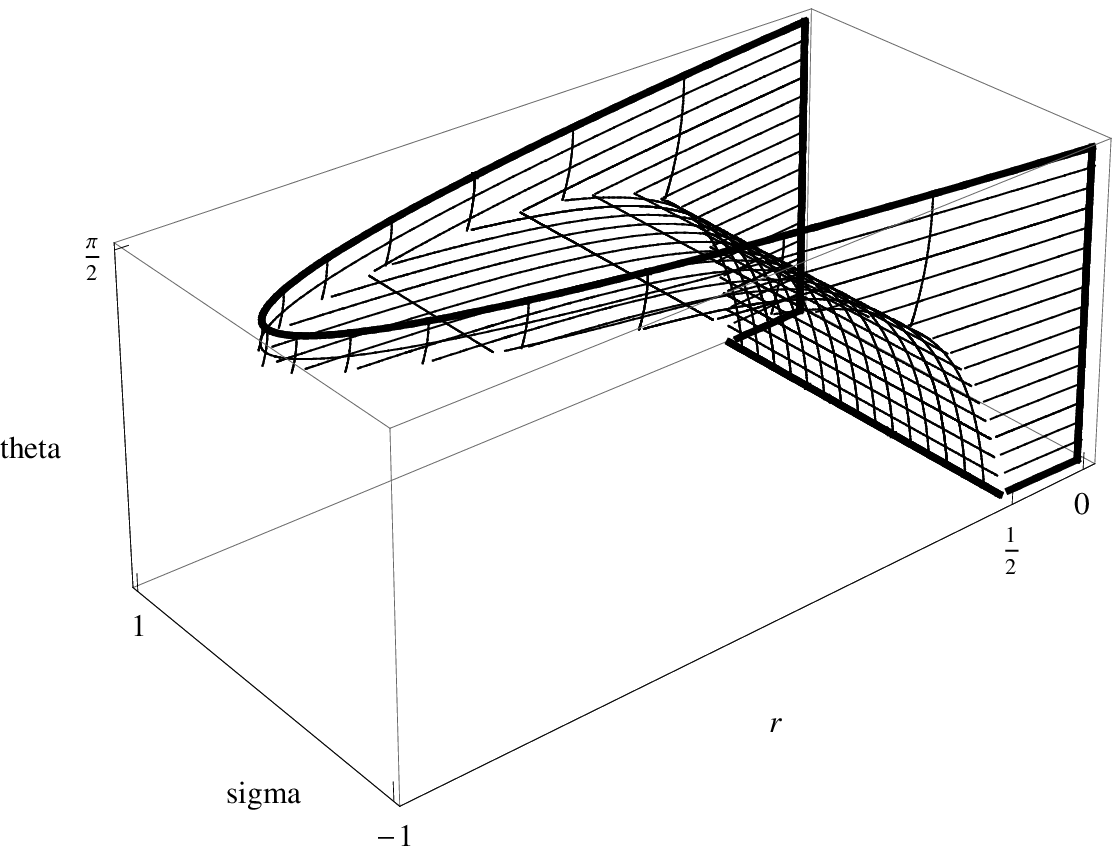}\label{Yspace}}\qquad\qquad
\subfigure[A schematic depiction of $\partial\mathcal{Y}_{\mathrm{IX}}^+$]{\includegraphics[width=0.3\textwidth]{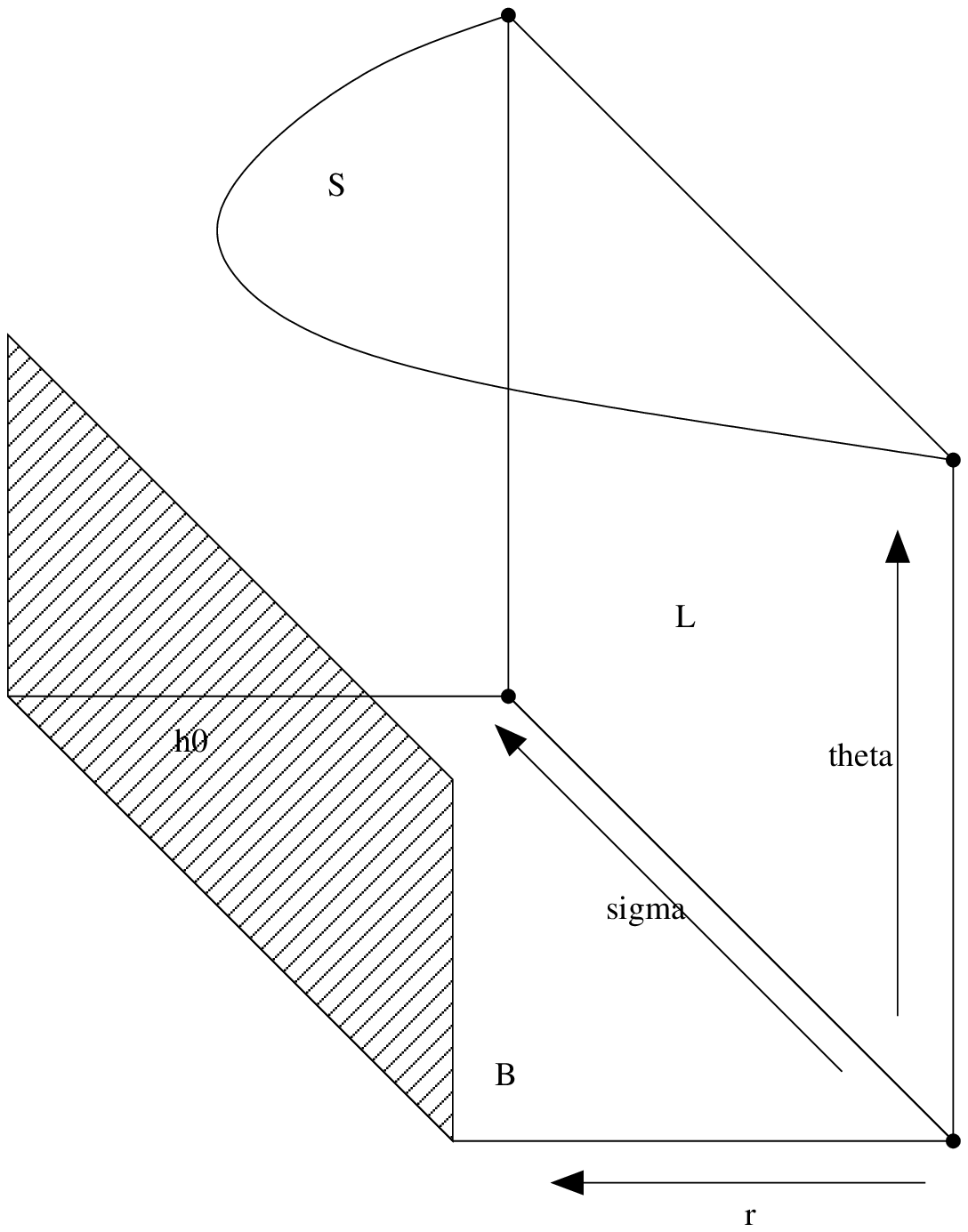}}
\caption{The space $\mathcal{Y}_{\mathrm{IX}}^+$ and its boundary. 
The various components of the boundary are defined in eqs.~\eqref{boundaryeqs}.}
\end{center}
\end{figure}

In the new variables, the dynamical system~\eqref{dynsyst} takes the form
\begin{subequations}\label{dynsyspolar}
\begin{align}
r^\prime & = 2\, r \left( H_D (q - H_D \Sigma_+) - 3 \Sigma_+ \sin^2 \vartheta \right), \\[0.5ex]
\label{polartheta}
\vartheta^\prime & = - 3 \Sigma_+ \sin (2\vartheta)\:, \\[0.5ex]
\label{polarSigma+}
\Sigma_+^\prime & = 
r \sin\vartheta -1 + (H_D -\Sigma_+)^2 + H_D \Sigma_+ (q- H_D\Sigma_+) + \Omega \,\big(w_2(\ell,s) -w_1(\ell,s) \big)\:,\\[0.8ex]
\ell'&=2 H_D \ell (1-\ell )\:.
\end{align}
where $q=2\Sigma_+^2+\frac{1}{2}(1+3w)\Omega$ and 
$\Omega =  1 - \Sigma_+^2 - \frac{1}{4}\, r \sin\vartheta$.
In addition, $H_D$ and $s$ are regarded as functions of $r$ and $\vartheta$ in~\eqref{dynsyspolar},
\begin{equation}\label{ho}
H_D  = \sqrt{1 -2 r \cos\vartheta} \;,\qquad
s=\frac{1}{2}\:\frac{\tan \vartheta}{1 + \tan \vartheta} \:.
\end{equation}
\end{subequations}
Note that these are smooth functions of $r$ and $\vartheta$ 
on the state space $\mathcal{Y}_{\mathrm{IX}}^+$. In particular,
\begin{equation*}
\frac{\partial s}{\partial \vartheta} = \frac{1}{\sin(2 \vartheta)} \,2 s (1- 2s)  = 
\frac{1}{2} \frac{1}{1+ \sin (2\vartheta)} \:,
\qquad\text{hence}\quad
\frac{1}{4} \leq \frac{\partial s}{\partial \vartheta} \leq \frac{1}{2} 
\quad\forall \vartheta \in [0,\frac{\pi}{2}]
\end{equation*}
and 
$\partial s/\partial \vartheta = \frac{1}{2}$ at $\vartheta = 0$ and $\vartheta = \frac{\pi}{2}$.

The flow of the system~\eqref{dynsyspolar} on the state space $\mathcal{Y}_{\mathrm{IX}}^+ \times (0,1)$ is
well defined in the past direction of time. This reflects the 
past invariance of the original domain $\mathcal{X}_{\mathrm{IX}}^+ \times (0,1)$.
Accordingly, an orbit of~\eqref{dynsyspolar} represents the expanding
phase of an LRS Bianchi type~IX solution with Vlasov matter;
conversely, the expanding phase of every LRS type~IX model is represented
by an orbit of~\eqref{dynsyspolar}.

In contrast to the system~\eqref{dynsyst}, 
the dynamical system~\eqref{dynsyspolar} 
on the state space $\mathcal{Y}_{\mathrm{IX}}^+ \times (0,1)$
\textit{admits a regular extension to the boundaries}. 
By the boundaries of~\eqref{Yss} we mean 
\begin{subequations}\label{boundaryeqs}
\begin{alignat}{2}
\label{r0}
& r = 0 : & \quad & \big\{ \vartheta \in\big(0,\frac{\pi}{2}\big), \Sigma_+ \in (-1,1) \big\} \:,\\
\label{vartheta0}
& \vartheta = 0 : & & \big\{ r \in \big(0,\frac{1}{2}\big), \Sigma_+ \in (-1,1) \big\} \:,\\
\label{varthetapi2}
& \vartheta = \frac{\pi}{2} : & &  \big\{ r \in \big(0,4(1-\Sigma_+^2)\big), \Sigma_+ \in (-1,1) \big\}\:, \\
\label{Omega0}
& \Omega = 0 : & & 
\Big\{ r = \min\Big[\frac{1}{2 \cos\vartheta},\frac{4(1-\Sigma_+^2)}{\sin\vartheta}\Big] = \frac{4(1-\Sigma_+^2)}{\sin\vartheta},
\vartheta \in(0,\frac{\pi}{2}), \Sigma_+ \in (-1,1) \Big\} \:,
\end{alignat}
\end{subequations}
and the closures of these invariant sets.
The set
\[
\Big\{ r = \min\Big[\frac{1}{2 \cos\vartheta},\frac{4(1-\Sigma_+^2)}{\sin\vartheta}\Big] = \frac{1}{2 \cos\vartheta},
\vartheta \in(0,\frac{\pi}{2}), \Sigma_+ \in (-1,1) \Big\} 
\]
is the preimage of the plane $H_D = 0$ under~\eqref{polartransf};
it is not an invariant set;
by Lemma~\ref{HDLemma} it is irrelevant for our considerations.
This exhausts the list of boundaries of~\eqref{Yss}.

\begin{Lemma}\label{r0thetapi2}
Let\/ $\mathrm{P}$ be an $\alpha$-limit point of an orbit of~\eqref{dynsyspolar}
(i.e., of a Bianchi type~IX solution with Vlasov matter).
Then 
\[
r_{| \mathrm{P}} = 0 \quad\text{\textnormal{or}}\quad
\vartheta_{| \mathrm{P}} = \frac{\pi}{2}\:.
\]
\end{Lemma}

\begin{proof}
We consider the function
\[
Y = H_D^6 \frac{\tan\vartheta}{r^3 \cos^3\vartheta}\:,
\]
which is smooth and positive on $\mathcal{Y}_{\mathrm{IX}}^+$
and on the boundary subset $\Omega = 0$, see~\eqref{Omega0}.
We find that
\begin{equation}\label{Ymon}
Y^{\prime} = {-12} Y (\Sigma_+^2 + \Omega) \:,\qquad
Y^{\prime\prime\prime}_{\quad|\Sigma_+ = 0,\Omega=0} = {-24} Y \big( 3 + H_D^2 \big) \:,
\end{equation}
hence $Y$ is strictly monotonically decreasing along every orbit.
This excludes the existence of \mbox{$\alpha$-limit} points 
in $\mathcal{Y}_{\mathrm{IX}}^+$ and on the boundary $\Omega = 0$.
Suppose that there exists an orbit that possesses an $\alpha$-limit point $\mathrm{P}$
with $\vartheta_{| \mathrm{P}} = 0$ (and $r_{| \mathrm{P}} > 0$).
Then there exists a sequence of times $(\tau_n)_{n\in\mathbb{N}}$ with
$\tau_n\rightarrow {-\infty}$ ($n\rightarrow \infty$) 
such that $Y(\tau_n) \rightarrow 0$ ($n\rightarrow \infty$) along the orbit; a contradiction to~\eqref{Ymon}.
\end{proof}

Lemma~\ref{r0thetapi2} suggests to study the dynamical system
that is induced by~\eqref{dynsyspolar} on the boundaries $r=0$ and $\vartheta = \frac{\pi}{2}$,
see~\eqref{r0} and~\eqref{varthetapi2}.

\textbf{The boundary subset $\bm{r = 0}$.}
The dynamical system~\eqref{dynsyspolar} induces the system
\begin{subequations}\label{Ibou}
\begin{align}
\vartheta^\prime & =- 3 \Sigma_+ \sin(2\vartheta)\:, \\
\Sigma_+^\prime & =-(1-\Sigma_+^2)\big[\Sigma_+ - \big(w_1(0,s) - w_2(0,s)\big)\big]\:,
\end{align}
\end{subequations}
on the boundary subset represented by $r=0$ and $\ell = 0$, see~\eqref{r0}. 
From~\eqref{wi} we see that
\begin{subequations}\label{wi0s}
\begin{align}
w_1(0,s) & = (1-2s)
\,\frac{\int f_0 \,u_1^2\, \left[(1-2s)u_1^2+s(u_2^2+u_3^2)\right]^{-1/2}\,du}%
{\int f_0\left[(1-2s)u_1^2+s (u_2^2+u_3^2)\right]^{1/2}du}\:,\\[1ex]
w_2(0,s) & = s\; \frac{\int f_0 \,u_2^2 \,
\left[(1-2s)u_1^2+s(u_2^2+u_3^2)\right]^{-1/2}\,du}%
{\int f_0\left[(1-2s)u_1^2+s (u_2^2+u_3^2)\right]^{1/2}du}\:,
\end{align}
\end{subequations}
which implies that $w = \textfrac{1}{3} \big(w_1 + 2 w_2\big) = \textfrac{1}{3}$. 
Recall that in the context of~\eqref{Ibou}, $s = s(\vartheta)$, see~\eqref{ho}.

It is straightforward to see that $w_1(0,s) - w_2(0,s) = 1- 3 w_2(0,s)$ is
a strictly monotonic function~\cite{RT}. 
Since $w_2(0,0) = 0$ and $w_2\big(0,\textfrac{1}{2}\big)=\textfrac{1}{2}$, 
there is a unique value $s_0\in (0,\frac{1}{2})$ such 
that $w_2(0,s_0)=\textfrac{1}{3}$. Let $\vartheta_0$ denote the (unique) value of $\vartheta$ 
associated with $s_0$ by~\eqref{ho}. 
We conclude that
there exists a unique fixed point $\mathrm{F}$ in the interior of 
the set $r=0$, $\ell = 0$:
\begin{itemize}
\item[$\mathrm{F}$] \quad The fixed point $\mathrm{F}$ is given by $r=0$, $\vartheta = \vartheta_0$, $\Sigma_+ = 0$, $\ell = 0$.
\end{itemize}

The remaining fixed points of the system~\eqref{Ibou} are located on
the boundary of the subset $r=0$; these are given by $\vartheta = 0$,
$\vartheta = \textfrac{\pi}{2}$, and $\Omega = 0$ (with $\Sigma_+ = \pm 1$).
\begin{itemize}
\item[$\mathrm{T}_\sharp$] \quad The \textit{Taub point\/} $\mathrm{T}_\sharp$ is given by 
$r=0$, $\vartheta = \textfrac{\pi}{2}$, $\Sigma_+ = {-1}$, $\ell = 0$.
\item[$\mathrm{Q}_\sharp$] \quad The \textit{non-flat LRS point\/} $\mathrm{Q}_\sharp$ is given by 
$r=0$, $\vartheta = \textfrac{\pi}{2}$, $\Sigma_+ = 1$, $\ell = 0$.
\item[$\mathrm{R}_\sharp$] \quad The point $\mathrm{R}_\sharp$ is given by 
$r=0$, $\vartheta = \textfrac{\pi}{2}$, $\Sigma_+ = \textfrac{1}{2}$, $\ell = 0$.
\item[$\mathrm{T}_\flat$] \quad The  \textit{Taub point\/} $\mathrm{T}_\flat$ is given by 
$r=0$, $\vartheta = 0$, $\Sigma_+ = {-1}$, $\ell = 0$.
\item[$\mathrm{Q}_\flat$] \quad The \textit{non-flat LRS point\/} $\mathrm{Q}_\flat$ is given by 
$r=0$, $\vartheta = 0$, $\Sigma_+ = 1$, $\ell = 0$.
\end{itemize}

The physical \textit{interpretation} of the fixed points is 
straightforward.
At the fixed point $\mathrm{F}$,
since $w_1 = w_2 = w_3 = w = \textfrac{1}{3}$, 
the principal pressures are identical and the matter is thus isotropic.
Since $\Sigma_+ = 0$, we have $k^1_{\ 1} = k^2_{\ 2} = k^3_{\ 3}$,
i.e., the geometry is isotropic as well.
Accordingly,
the fixed point $\mathrm{F}$ represents the 
Friedmann-Lema\^itre-Robertson-Walker model, 
\begin{equation}\label{frw}
g_{11}=g_{22}=g_{33}=a\,t\qquad (a>0)\:,
\end{equation}
i.e., a spatially flat and isotropic solution of the Einstein-Vlasov system.

The fixed points $\mathrm{T}_\flat$ and $\mathrm{T}_\sharp$, as well as the orbit connecting
these two fixed points, cf.~Fig.~\ref{Ifig}, represent solutions of the Taub class, i.e.,~\eqref{taub}.
During the time that an orbit of~\eqref{dynsyst} is close to $\mathrm{T}_\flat$, $\mathrm{T}_\sharp$, or 
$\mathrm{T}_\flat \rightarrow\mathrm{T}_\sharp$, it represents a
metric~\eqref{metric} whose components are well approximated by~\eqref{taub}.
(To see this we simply use the variable transformations~\eqref{newvar} and~\eqref{polartransf}.)
Likewise, while an orbit of~\eqref{dynsyst} is close to $\mathrm{Q}_\flat$, $\mathrm{Q}_\sharp$, or 
$\mathrm{Q}_\flat \leftarrow\mathrm{Q}_\sharp$, it represents a
metric~\eqref{metric} whose components are well approximated by~\eqref{solQ}.

The solution associated to the fixed point $\mathrm{R}_\sharp$ is of type~I as well.
(This fixed point has already been found in~\cite{RT}, where it was denoted by $P_3$.)
The metric is
\[
g_{11}=a\:,\quad g_{22}=g_{33}=b\,t^{4/3} \qquad (a,b>0)\:;
\]
we refer to~\cite{CH1}. 
Note that this a non-vacuum solution and that the matter is anisotropic; hence
$\mathrm{R}_\sharp$ does not correspond to 
a perfect fluid solution.

\vspace{2ex}

\begin{Lemma}\label{ILemma}
The system~\eqref{Ibou} on the closure of the subset\/ $r = 0$ (and $\ell = 0$) gives
rise to the flow depicted in Fig.~\ref{Ifig}.
\end{Lemma}

\begin{proof}
Consider the function 
\[
Z=(1-\Sigma_+^2)^{-1} \big((1-2s)s^2)\big)^{-1/6}\int f_0 \left[(1-2s)u_1^2+s\big(u_2^2+u_3^2\big)\right]^{1/2}dv\:,
\]
which satisfies
\[
Z'=-2 Z \Sigma_+^2\:,\qquad 
Z'''_{\;\:|\Sigma_+=0}=-36 Z \big( w_2(0,s)-\textfrac{1}{3}\big)^2 \:.
\]
Since $Z$ is strictly monotonically decreasing except at $\mathrm{F}$,
we conclude that $\mathrm{F}$ is the $\omega$-limit set of every interior orbit, 
while the $\alpha$-limit set is contained on the boundary.
The straightforward analysis of the flow on the one-dimensional boundaries leads to the claim of
the lemma.
\end{proof}

\begin{figure}[Ht!]
\begin{center}
\psfrag{T}[cc][cc][1][0]{$\mathrm{T}_\flat$}
\psfrag{T*}[cc][cc][1][0]{$\mathrm{T}_\sharp$}
\psfrag{Q}[cc][cc][1][0]{$\mathrm{Q}_\flat$}
\psfrag{Q*}[cc][cc][1][0]{$\mathrm{Q}_\sharp$}
\psfrag{D*}[cc][cr][1][0]{$\mathrm{R}_\sharp$}
\psfrag{F}[cc][cc][1][0]{$\mathrm{F}$}
\psfrag{th}[cc][cc][0.7][0]{$\vartheta$}
\psfrag{sig}[ll][rr][0.7][0]{$\Sigma_+$}
\includegraphics[width=0.5\textwidth]{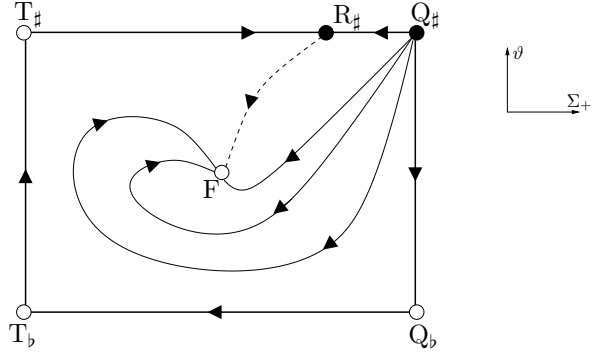}
\end{center}
\caption{The flow induced on the two-dimensional boundary subset determined by $r =0$ 
and $\ell = 0$. The color-coding of the fixed points has the following meaning: A 
fixed point depicted in black (resp.\ white) is an attractor (resp.\ repeller) 
of orbits in a transversal direction (within the set $\ell = 0$).}
\label{Ifig}
\end{figure}

\textbf{The boundary subset $\bm{\vartheta = \textfrac{\pi}{2}}$.}
The dynamical system~\eqref{dynsyspolar} induces the system
\begin{subequations}\label{flowSsharp}
\begin{align}
r'&=2r\big(\Sigma_+^2-4\Sigma_++1-\frac{r}{4}\big)\:, \\
\Sigma'_+&=\frac{r}{2} (2-\Sigma_+)- \big(1-\Sigma_+^2-\frac{r}{4}\big) \big(\Sigma_+-\frac{1}{2}\big)\:,
\end{align}
\end{subequations}
on the boundary subset represented by $\vartheta=\textfrac{\pi}{2}$ and $\ell = 0$, see~\eqref{varthetapi2}. 
Here we have used that $s=\textfrac{1}{2}$ when $\vartheta = \textfrac{\pi}{2}$, see~\eqref{ho}, which implies that 
$w_1(0,s) = 0$ and $w_2(0,s) = \textfrac{1}{2}$, see~\eqref{wi0s}.

The one-dimensional boundary of the system~\eqref{flowSsharp} consists of a part where $\Omega = 0$
and a part $r = 0$; on the latter, the systems~\eqref{Ibou} and~\eqref{flowSsharp} intersect.

\begin{Lemma}\label{IILemma}
The system~\eqref{flowSsharp} on the closure of the subset\/ $\vartheta=\textfrac{\pi}{2}$ (and $\ell = 0$) 
gives rise to the flow depicted in Fig.~\ref{IIfig}.
\end{Lemma}

\begin{proof}
The absence of fixed points, periodic orbits and heteroclinic cycles in the interior of this set
implies, by the Poincar\'e-Bendixson theorem~\cite{perko}, 
that the $\alpha$- and $\omega$-limit points of orbits must be located on the boundary. 
A local dynamical systems analysis then yields the claim of the lemma.
\end{proof}

\begin{figure}[Ht]
\begin{center}
\psfrag{s+}[cc][cc][0.8][0]{$\Sigma_+$}
\psfrag{T}[cc][cc][1][0]{$\mathrm{T}_\sharp$}
\psfrag{Q}[cc][cc][1][0]{$\mathrm{Q}_\sharp$}
\psfrag{D}[cc][cc][1][0]{$\mathrm{R}_\sharp$}
\psfrag{r}[cc][cc][0.8][0]{$r$}
\includegraphics[width=0.7\textwidth]{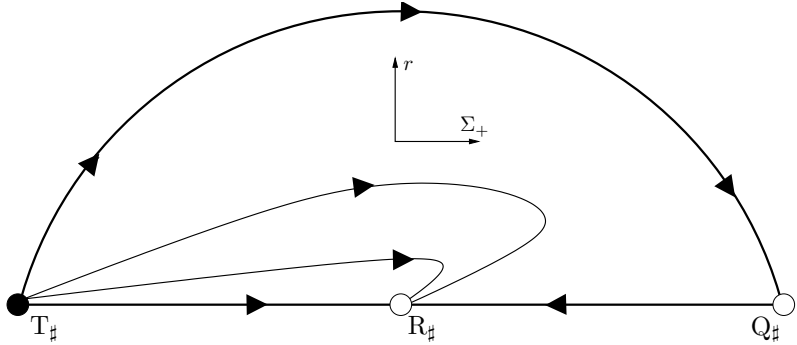}
\end{center}
\caption{The flow induced by~\eqref{dynsyspolar} on the two-dimensional boundary subset determined by $\vartheta =\frac{\pi}{2}$ 
and $\ell = 0$.}
\label{IIfig}
\end{figure}

This completes our analysis of the boundary subsets.
Let us return to the study of~\eqref{dynsyspolar}.

\begin{Lemma}\label{FLemma}
There exists a one-parameter family of orbits of~\eqref{dynsyspolar}
whose $\alpha$-limit set is the fixed point\/ $\mathrm{F}$.
\end{Lemma}

\begin{proof}
A simple calculation shows that the unstable manifold of $\mathrm{F}$ is
two-dimensional. (See Lemma~\ref{ILemma} for the stable manifold.)
\end{proof}

\begin{Lemma}\label{nor0Lemma}
Assume that a point\/ $\mathrm{P}$ 
of the (interior of the) boundary subset $r = 0$, $\ell = 0$, see~\eqref{r0},
is an $\alpha$-limit point of 
an orbit of~\eqref{dynsyspolar}.
Then\/ $\mathrm{P} = \mathrm{F}$.
\end{Lemma}

\begin{proof}
Let the orbit under consideration be denoted by $\gamma$. 
Assume that $\mathrm{P} \neq \mathrm{F}$. 
Together with $\mathrm{P}$, the entire orbit through $\mathrm{P}$ and
its $\omega$-limit point $\mathrm{F}$ must be in the $\alpha$-limit set of $\gamma$; see Lemma~\ref{ILemma}.
Since $\mathrm{F}$ is a saddle point, see Lemmas~\ref{ILemma} and~\ref{FLemma}, there is a point $\hat{\mathrm{P}}$
on the unstable manifold of $\mathrm{F}$ that is contained in $\alpha(\gamma)$. 
(Lemma~\ref{ellLemma} enforces $\hat{\mathrm{P}}$ to be located on $\ell = 0$.)
Let $(\tau_n)_{n\in\mathbb{N}}$ with $\tau_n\rightarrow {-\infty}$ ($n\rightarrow \infty$) 
denote a sequence of times such that $\gamma(\tau_n) \rightarrow \hat{\mathrm{P}}$
($n\rightarrow \infty$); by assumption, none of the points $\gamma(\tau_n)$ are contained on the unstable manifold
of $\mathrm{F}$. The orbit through $\hat{\mathrm{P}}$ represents a solution
of the Einstein-Vlasov equations that satisfies Lemma~\ref{HDLemma}.
By continuous dependence on initial data we can thus construct a sequence of times $\hat{\tau}_n$  
with $\hat{\tau}_n\rightarrow {-\infty}$ ($n\rightarrow \infty$) such that
$H_D(\hat{\tau}_n) \rightarrow 0$ ($n\rightarrow \infty$); 
this is in contradiction to Lemma~\ref{HDLemma}.
\end{proof}

\begin{Lemma}\label{novarthetapi2Lemma}
There does not exist any point 
in the (interior of the) boundary subset $\vartheta = \textfrac{\pi}{2}$, $\ell = 0$, see~\eqref{varthetapi2},
that is an $\alpha$-limit point of 
an orbit of~\eqref{dynsyspolar}.
\end{Lemma}

\begin{proof}
Assume that there exists an orbit $\gamma$ such that $\mathrm{P} \in \alpha(\gamma)$,
where $\mathrm{P}$ lies on the (interior of the) boundary subset $\vartheta = \textfrac{\pi}{2}$, see~\eqref{varthetapi2}.
Together with $\mathrm{P}$, the entire orbit through $\mathrm{P}$ and
its $\omega$-limit point $\mathrm{R}_\sharp$ must be in the $\alpha$-limit set of $\gamma$; see Lemma~\ref{IILemma}.
The fixed point $\mathrm{R}_\sharp$ is a saddle point, see Lemmas~\ref{ILemma} and~\ref{IILemma}; hence 
there is a point $\hat{\mathrm{P}}$ on the unstable manifold of $\mathrm{R}_\sharp$ that is contained in $\alpha(\gamma)$. 
The unstable manifold of $\mathrm{R}_\sharp$ is two-dimensional; a particular orbit contained in it is the orbit
$\mathrm{R}_\sharp \rightarrow \mathrm{F}$ depicted in Fig.~\ref{Ifig}; the orthogonal unstable direction is the $\ell$-direction.
Lemma~\ref{ellLemma} implies that $\ell_{|\hat{\mathrm{P}}} = 0$, hence $\hat{\mathrm{P}}$ is located on the orbit
$\mathrm{R}_\sharp \rightarrow \mathrm{F}$ of Fig.~\ref{Ifig}. However, this is a contradiction to Lemma~\ref{nor0Lemma}.
\end{proof}

\begin{Lemma}\label{notonelineLemma}
There does not exist any point 
on the (interior of the) line\/ $\mathrm{T}_\sharp \rightarrow \mathrm{R}_\sharp \leftarrow \mathrm{Q}_\sharp$,
see Figs.~\ref{Ifig} and~\ref{IIfig},
that is an $\alpha$-limit point of 
an orbit of~\eqref{dynsyspolar}.
\end{Lemma}

\begin{proof} The proof is analogous to the proof of Lemma~\ref{novarthetapi2Lemma}.
\end{proof}

Summarizing the statements of the previous lemmas we obtain the main theorem
of this work. 
We give a formal statement of the theorem using
the developed framework; the less formal interpretation of the theorem
is presented in the concluding remarks.

\begin{Theorem}\label{BianchiIXtheo} 
Consider an orbit of~\eqref{dynsyspolar} representing 
(the expanding phase of) a Bianchi type~IX solution with Vlasov matter.
Then the $\alpha$-limit set of this orbit is either the fixed point $\mathrm{F}$,
which is the non-generic case, or it is the heteroclinic cycle 
of orbits connecting the four fixed points $\mathrm{T}_\flat$, $\mathrm{T}_\sharp$,
$\mathrm{Q}_\sharp$, $\mathrm{Q}_\flat$, see Fig.~\ref{BianchiIXfig}; this is the generic case.
\end{Theorem}

\begin{proof}
By Lemma~\ref{FLemma} there exists a non-generic family of orbits whose $\alpha$-limit
set is $\mathrm{F}$. 
Consider an orbit that is not a member of this family.
Lemma~\ref{ellLemma} implies that the $\alpha$-limit set must be located on the
boundary $\ell = 0$ of the state space of~\eqref{dynsyspolar}.
Lemma~\ref{r0thetapi2} enforces $\alpha$-limit points to lie on the closures
of the boundary subsets $r = 0$ and $\vartheta = \textfrac{\pi}{2}$.
However, Lemmas~\ref{nor0Lemma} and~\ref{novarthetapi2Lemma} imply that points in the interior of
these subsets are excluded. Finally, by Lemma~\ref{notonelineLemma}
we find that the set that remains as the possible location of $\alpha$-limit points
is a one-dimensional set, the heteroclinic cycle connecting
the fixed points $\mathrm{T}_\flat$, $\mathrm{T}_\sharp$, $\mathrm{Q}_\sharp$, $\mathrm{Q}_\flat$.
\end{proof}

\begin{figure}[Ht!]
\begin{center}
\psfrag{tf}[cc][cc][0.7][0]{$\mathrm{T}_\flat$}
\psfrag{ts}[cc][cc][0.7][0]{$\mathrm{T}_\sharp$}
\psfrag{qf}[cc][cc][0.7][0]{$\mathrm{Q}_\flat$}
\psfrag{qs}[cc][cc][0.7][0]{$\mathrm{Q}_\sharp$}
\psfrag{f}[cc][cc][0.7][0]{$\mathrm{F}$}
\psfrag{c}[cc][cc][0.7][0]{$\mathrm{C}_\sharp$}
\psfrag{d}[cc][cc][0.7][0]{$\mathrm{R}_\sharp$}
\psfrag{r}[cc][cc][0.7][0]{$\mathrm{R}_\flat$}
\psfrag{sig+}[cc][cc][0.7][0]{$\Sigma_+$}
\psfrag{th}[cc][cc][0.7][0]{$\vartheta$}
\psfrag{rc}[cc][cc][0.7][0]{$r$}
\psfrag{h0}[cc][cc][0.7][-45]{$H_D=0$}
\includegraphics[width=0.5\textwidth]{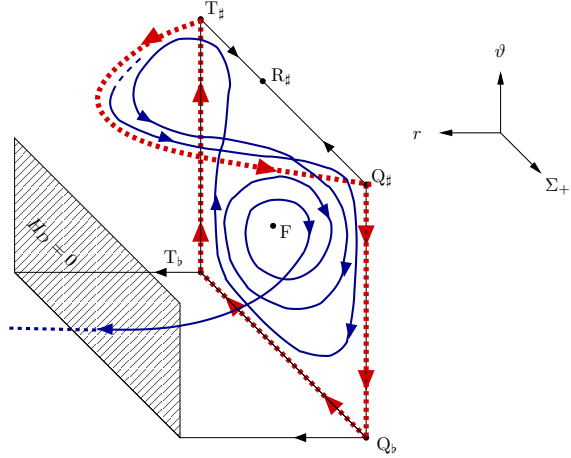}
\end{center}
\caption{The projection of a generic orbit (blue line) of the dynamical system~\eqref{dynsyspolar}
to the set $\ell =0$.
Dashed red lines 
are orbits that lie on the boundary of the state space 
and connect to form a heteroclinic cycle. 
There exist non-generic orbits that are asymptotic to the fixed point $\mathrm{F}$ (not depicted).}
\label{BianchiIXfig}
\end{figure}


\section{Concluding remarks}

\label{conc}


Let 
us give an interpretation of the main result, Theorem~\ref{BianchiIXtheo}, in informal terms:
Every LRS Bianchi type~IX solution
with collisionless matter (i.e., solution of the Einstein-Vlasov system)
possesses an initial singularity (which is chosen to be at $t = 0$).
The behavior as $t\rightarrow 0$ of generic solutions, where by `generic solutions' we mean a family of solutions that
corresponds to a set of initial data of full measure, 
is characterized by oscillations between the Taub solution~\eqref{taub} and the 
non-flat LRS solution~\eqref{solQ}: There exists an infinite sequence
of (increasingly small) time intervals such that the components of the 
metric~\eqref{metric} are well approximated by~\eqref{taub},
and a sequence of time intervals where~\eqref{solQ} yields a good approximation;
the accuracy of the approximation increases as $t\rightarrow 0$.
During the transitions between~\eqref{taub} and~\eqref{solQ} the metric
takes an entirely different form; we merely note that
each transition from the non-flat LRS solution to the Taub solution (which correspond to 
an orbit of~\eqref{dynsyspolar} closely following the orbit $\mathrm{Q}_\flat \rightarrow \mathrm{T}_\flat$
in Fig.~\ref{BianchiIXfig}) is characterized by the fact that the variable $\Omega$
attains values close to $1$. This means that 
the influence of the matter on the (asymptotic) dynamics cannot be neglected.
The oscillatory behavior toward the singularity of 
generic LRS Bianchi type~IX cosmological models with collisionless matter 
is thus qualitatively different from that of perfect fluid 
models with the same symmetry. In fact, in the LRS case, perfect fluid cosmologies of 
Bianchi type~IX are asymptotic to the Taub solution~\eqref{taub}, which 
is also the behavior of LRS Bianchi type~IX {\it vacuum} cosmological models. 
Therefore, as the `structure' of the singularity is concerned, 
the assumption of perfect fluid matter does not change the vacuum behavior
(`matter does not matter'), 
while `collisionless matter matters'; Vlasov matter
has an important effect on the structure of the singularity.

The occurrence of oscillatory dynamics toward the singularity 
is intimately connected with the (in)stability of the Kasner
fixed points (i.e., the Taub points and the non-flat
LRS points in the context of the present work).
While stable in the context of the simplest cosmological models,
which are characterized by a small number of (true) 
degrees of freedom and by the isotropy of the matter model,
the Kasner points run the risk of loosing their stability
when degrees of freedom are added to the problem 
or when the matter model becomes anisotropic.
Loss of stability is then a possible cause for
the existence of heteroclinic structures which in turn 
induce oscillatory behavior of cosmological models.
The prime example is the
transition from (non-LRS) Bianchi type~$\mathrm{VI}_0$ and $\mathrm{VII}_0$ models
to (non-LRS) type~VIII and~IX models in the vacuum case.
The dimension of the state space increases, the Kasner points
lose their stability properties, and the dynamics
of (vacuum) cosmological models become the oscillatory
``Mixmaster'' dynamics.
Similarly in spirit, in the context of LRS Bianchi models with Vlasov matter,
the dimension of the state space increases when one
goes over from LRS type~I to LRS type~II models.
The loss of stability induces oscillatory behavior
of LRS Bianchi type~II solutions with Vlasov matter, see~\cite{RT,RU}.
However, 
already a change
of matter model alone (where the number of degrees of
freedom of the problem are unaffected) can change the stability properties
of the Kasner points: For instance, while the asymptotic dynamics of
LRS Bianchi type~I solution with Vlasov matter are ``monotone'', one
observes oscillatory behavior for LRS Bianchi type~I solutions 
with elastic matter~\cite{CH5}.  
The present work provides another example: The state spaces
and reduced dynamical systems that represent the dynamics
of LRS type~IX perfect fluid models 
and LRS type~IX models with massless Vlasov matter
are of the same dimension; the same is true
for the state space of LRS type~IX perfect fluid models 
with non-linear equations of state and the state space
of type~IX models with massive Vlasov matter
considered in this paper.
Despite this fact, the qualitative dynamics 
are radically different, which is ultimately because the Kasner
points exhibit different stability properties.
In brief: ``Matter matters''.

Note that by the symmetry of the problem, the behavior
of generic solutions toward the final singularity (see Lemma~\ref{HDLemma})
mirrors the behavior toward the initial singularity.
Furthermore, we note 
there exists a non-generic family of LRS Bianchi type~IX solutions 
with collisionless matter 
that isotropize toward the initial singularity, see Lemma~\ref{FLemma}. 
These solutions are asymptotic 
to the isotropic Friedmann-Lema\^itre-Robertson-Walker model~\eqref{frw}, i.e.,
$g_{i j}(t) \sim t \delta_{i j}$ as $t\rightarrow 0$. 
Although these non-generic solutions do not influence the asymptotic behavior 
of generic ones, they may affect the intermediate behavior 
of generic solutions and thus play a relevant role in physics. 
We refer to~\cite{WE} for a discussion on this interesting topic.

In this paper we have shown that collisionless matter 
has an important effect on the structure of the singularity
already in the context of solutions of the Einstein equations
with high symmetries (LRS Bianchi type~IX).
It is to be expected that this type of effect becomes
even more pronounced when the degree of symmetry is reduced.
In particular, an exciting question is in which manner collisionless matter
will influence the so-called Mixmaster behavior of
generic (non-LRS) vacuum Bianchi type~IX models.
Na\"ively, one might expect that the two types of oscillations,
the Mixmaster oscillations (which are induced, in a manner of speaking,
by gravity itself) and the oscillations caused by the anisotropy
of the matter model 
`intertwine' to yield a intricate oscillatory structure.
Unfortunately, it is difficult to give a
well formulated conjecture, the reason being that 
it is not known at present how to extend the dynamical 
systems method to the full (non-LRS) Bianchi type~IX case.

\vspace{0.5cm}

\noindent {\bf Acknowledgments:} 
The authors would like to thank an anonymous referee for valuable suggestions.
S.\ C.\ is supported by  Ministerio Ciencia e Innovaci\'on, Spain (Project MTM2008-05271).

\end{document}